\documentclass[12pt]{article}
\newtheorem{proposition}{Proposition}

\newtheorem{lemma}{Lemma}

\newtheorem{example}{Example}
\newenvironment{proof}{\noindent{\bf Proof:}}{\hfill\fbox{}\vspace*{1mm}}
\usepackage[usenames]{color}
\usepackage{graphicx}
\usepackage{amsmath}
\usepackage{amssymb}

\RequirePackage[normalem]{ulem} 
\providecommand{\DIFdeltex}[1]{{\protect\color{red}\sout{#1}}}                      
\usepackage{xcolor}
\usepackage{changebar}
\newif\ifdiff
\difftrue
\ifdiff
  \newcommand{\del}[1]{\DIFdeltex{#1}}

\else
  \newcommand{\del}[1]{}

\fi

\begin{document}
\title{\bf On Modeling Economic Default Time : A Reduced-Form Model Approach}
\author{
Jia-Wen Gu
\thanks{Advanced Modeling and Applied Computing Laboratory,
Department of Mathematics, The University of Hong Kong,
Pokfulam Road, Hong Kong. Email:jwgu.hku@gmail.com.
}
\and
Bo Jiang
\thanks{Advanced Modeling and Applied Computing Laboratory,
Department of Mathematics, The University of Hong Kong,
Pokfulam Road, Hong Kong. Email:sheilajiangbo@gmail.com.
}
\and Wai-Ki Ching
\thanks{Advanced Modeling and Applied Computing Laboratory, Department of Mathematics, The University of Hong Kong, Pokfulam
Road, Hong Kong.
E-mail: wching@hku.hk. Research supported in
part by GRF Grants, Hung Hing Ying Physical Research Grant and
HKU CRCG Grants.}
\and Harry Zheng
\thanks{ Department of Mathematics,
Imperial College, London, SW7 2AZ, UK. Email: h.zheng@imperial.ac.uk.}
}

\maketitle

\abstract{In the aftermath of the global financial crisis,
much attention has been paid to investigating the
appropriateness of the current practice of default risk modeling
in banking, finance and insurance industries.
A recent empirical study by Guo et al. (2008) \cite{Guo1} 
shows that the time difference between
the economic and recorded default dates has a significant impact on
recovery rate estimates.
Guo et al. (2011) \cite{Guo2} develop a
theoretical structural firm asset value model for a firm
default process that embeds the distinction of these two default times.
To be more consistent with the practice, in this paper,
we assume the market participants cannot observe the firm asset value directly and developed a reduced-form model to characterize the economic and recorded default times.
We derive the probability distribution of these two default times.
The numerical study on the difference between these two shows that
our proposed model can both capture the features and fit the empirical data.
}


\section{Introduction}

Modeling default risk has long been an important problem in
both theory and practice of banking and finance.
Popular credit risk models currently used have their
origins in two major classes of models.
The first class of models was pioneered by Black and Scholes (1973) \cite{BS}
and  Merton (1974) \cite{Merton} and is called 
the structural firm value model.
The basic idea of the model is to describe explicitly the
relationship between the asset value and the default of a firm.
More specifically, the default of the firm
is triggered by the event that the asset value
of the firm falls below a certain threshold
level related to the liabilities of the firm.
The structural firm value model provides
the theoretical basis for the commercial KMV model
which has been widely used for default risk model
in the financial industry.
The second class of models was developed by Jarrow and Turnbull (1995) \cite{JT}
and Madan and Unal (1998) \cite{MU}
and is called the reduced-form credit risk model.
The basic idea of the model is to consider defaults as exogenous events
and to model their occurrences by using Poisson processes and their variants.

A recent empirical study by Guo, Jarrow and Lin (2008) \cite{Guo1}
on the time-series behavior of market debt prices around the
recorded default date reveals the fact that the market anticipates the default event
well before default is recorded.
Their statistical analysis shows that the time span between
the economic and recorded default dates has a significant impact on
recovery rate estimates and is important
to obtaining unbiased estimates
for defaultable bond prices.
Guo et al. (2011) \cite{Guo2} develop a
theoretical structural firm asset value model for a firm default
process that embeds a distinction between an economic 
and a recorded default time and study the probability 
distributions of the economic and recorded default times.

In this paper, to be more consistent with the market practice,
we assume that the market participants cannot observe the 
firm asset value directly,
instead, they are aware of the firm's operation state.
The firm's state process is characterized by 
a continuous-time Markov chain with stochastic transition rates.
By this assumption, our proposed model,
{different from the one proposed by
Guo et al. (2011) \cite{Guo2},  is a ``reduced-form'' model.
Under this framework, the economic and recorded default time is defined in a similar manner as in Guo et al. (2011) \cite{Guo2}.
We derive the probability law of the economic and recorded default time.
Numerical study reveals that our proposed models can better capture the features
given by empirical study in Guo et al. (2008) \cite{Guo1}.

The rest of the paper is organized as follows.
Section 2 provides a review on
Guo et al.'s structural firm asset value model \cite{Guo2}.
Section 3 gives the construction of our proposed reduced-form model.
Section 4 presents the main results of this paper concerning the
distribution of economic and recorded default time.
Section 5 provides the numerical illustrations on the computation of economic and recorded default time distribution.
Section 6 then concludes the paper.

\section{Literature Review}

Guo et al. (2008) \cite{Guo1} show that identifying the ``economic'' default date, as distinct from the recorded
default date, is crucial to obtaining unbiased recovery estimates. For most debt issues,
the economic default date occurs far in advance of the reported default date. An
implication is that the standard industry practice of using 30-day post default
prices to compute recovery rate yields biased estimates.
This result, unfortunately, reveals that the empirical studies investigating the
economic characteristics of industry based recovery rates are using biased data.
Hence, the study of the economic default date is essential and important.

To be more specific, Guo et al. (2008) \cite{Guo1}
proposed a recovery rate model which fits the stressed bond prices well
with an average pricing error of less than one basis point. 
In their model, the ``modified recovery rate'' process
is defined to price the stressed bonds as follows:
$$
R_s=\delta_s e^{-\int_{\tau_e}^s r_u du}, \  s > \tau_e
$$
where $\delta_s$ denotes the recovery rate process
and  $\tau_e$ is the economic default time.
We remark that $R_s$ implicitly depends on the economic default time.

In Guo et al.'s model \cite{Guo2},
for a given a filtered probability space
$(\Omega, \mathcal{F}, \mathcal{F}_t, P)$
that satisfies the usual conditions,
the value of the firm $S = (S_t)_{t \geq 0}$ follows a geometric L$\acute{e}$vy process together with its natural filtration $\mathcal{F}_t$.
The firm needs to make debt repayments at a predetermined (deterministic) set of discrete times, denoted by $N_1,N_2, N_3, \ldots,$.
For simplicity, let $N_k = kN$ for a fixed $N > 0$, at time $N_k$, the
amount of debt in the firm is $D_k$.
For simplicity, we assume that $D_k = D$ is constant over time.
Consistent with a structural model, the recorded default time $\tau_r$ is the first time that the firm is unable to make a debt repayment, i.e.,
$$
\tau_r= \inf \{N_k: S _k \leq D\}
$$
while economic default time to be the last time, before the onset
of recorded default, when the firm is able to make a debt repayment, i.e.,
$$
\tau_e= \sup\{t \in [\tau_r-N, \tau_r]: S_t \geq D \}.
$$
The following proposition, given by Guo et al. (2011) \cite{Guo2},
characterizes the distribution of the important quantity ($\tau_r-\tau_e$),
the time lap between the recorded default time and the economic default time.

\begin{proposition}\label{Prop1}
(Guo, Jarrow and Larrard (2011) \cite{Guo2})
Assume that $S=(S_t, t \geq 0)$ is a geometric spectrally positive L$\acute{e}$vy process, then
$$
P_x(\tau_r-\tau_e \in ds)=\int_D^{\infty} \sum_{n=1}^{\infty}\psi(u,s) u_n(x) P_x(S_{(n-1)N} \in du \mid \tau_r =nN)
$$
where $u_n(x) =P_x(\tau_r =nN)$ and
$$
\psi(x,s)=\int_0^N P_{(u,D)}(\tau_r-\tau_e \in ds \mid \tau_r=N) P_x(H_D \in du)
$$
where  $H_D=\inf\{t: S_t \leq D\}$
and $P_{(u, D)}$ denotes the distribution of $S$ starting from $D$ at time $t=u$.
\end{proposition}
Suppose $(S_t, t \geq 0)$ is a geometric Brownian motion with zero drift, i.e.,
$$
S_t=\exp \left(\mu W_t- \frac{\mu^2 t}{2}\right)
$$
under the risk neutral measure with $W_t$ being a standard Brownian motion,
then we have
$$
P_{(u,D)}(\tau_r-\tau_e \in ds \mid \tau_r=N)
=\frac{ds}{\pi \sqrt{s(N-u-s)}}\phi \left(\frac{\mu}{2}\sqrt{N-u-s}\right),
$$
with
$$
\phi(a)=\int_0^{\infty} dt e^{-t} \cosh(a\sqrt{2t}).
$$
Therefore, the distribution of ($\tau_r-\tau_e$) is a mixture of arcsine law.
From the empirical study by Guo, Jarrow and Lin (2008) \cite{Guo1},
the density of time difference between the economic and the recorded default
has a ``$U$-shape'' in the time interval $[0,N]$,
while this feature can be well captured by the Arcsine law.

\section{The Reduced-Form Model}

We present our proposed reduced-form model in this section.
The distinction of the economic and recorded default time is also embeded.
We begin with a complete probability space $(\Omega, \mathcal{F}, \mathcal{F}_t, P)$.
Under this probability space, we are given a stochastic process $(X_t)_{t \geq 0}$,
right-continuous with left limits, representing the macroeconomic environment common factor.
We consider a firm with $K$ states, i.e., $1, 2, \ldots, K$,
where state $K$ represents the default state.
Let stochastic process $(S_t)_{t \geq 0}$ denotes the state process of the given firm and
we assume that $(S_t)_{t \geq 0}$ is a continuous-time Markov chain with stochastic transition rates,
i.e., $\lambda_{i,j}(X_s)$, where each $\lambda_{i,j}$ is
a bounded continuous function defined  on ${\mathbb R}$.
Heuristically, one can think of, $\lambda_{i,j}(X_s) \Delta t$ as the probability that a firm in state $i$ will
jump to state $j$ within the (small) time interval $\Delta t$.
With these notations,
the transition rate depends on the stochastic process $(X_s)_{s \geq 0}$ characterizing the common factor.
Let
$$
\lambda_i(X_s)=\sum_{k \neq i}\lambda_{i,k}(X_s), \quad i=1, 2, \ldots, K.
$$
Here $\lambda_i(X_s) \Delta t$ is the probability that a firm in state $i$ will
jump to different states within the (small) time interval $\Delta t$.

Here we redefine the economic and recored default time
under the given framework.
First, we assume the firm has to make certain required payment
at some fixed time, i.e.,
$0=N_0, N_1, \ldots, N_i, \ldots$.
For simplicity, we assume that the $N_i= i N$.
If the firm is in the ``default'' state at the payment date,
its payment will be missed.
The recorded default time $\tau_r$ is defined to be
$$\tau_r = \inf \{N_i: S_{N_i} = K\}$$
while the economic default time is defined to be
$$\tau_e = \sup \{t \leq \tau_r: S_t \neq K\}.$$
The information set available to the market participants up to time $t$ is then given by $$\mathcal{F}_t=\sigma(X_s, S_s, 0 \leq s \leq t).$$
For the ease of discussion, we also define
$$
\mathcal{G}_t=\sigma(X_s: 0 \leq s \leq t).
$$

\section{The Distribution of the Economic Default Time $\tau_e$}

In this section, we focus on finding the distributions of $\tau_r$ and $\tau_e$.
There are two cases to be discussed: constant transition rates and stochastic transition rates.
We begin with the following proposition which gives the probability law of the two random variables.

\begin{proposition}\label{Prop2}
For a non-negative integer $i$, we have
\begin{equation}\label{tau_e}
\begin{array}{lll}
 &&P(\tau_e \in (N_i, N_i +t] \mid \mathcal{G_{\infty}})\\
&=&\left(\displaystyle \prod_{j=0}^{i-1} P^{**}_X(N_j, N_{j+1})\cdot P^*_X(N_i, N_i +t)\right)_{S_0, K} \exp\left\{ -\int_{N_i+t}^{N_{i+1}} \lambda_K(X_u) du\right\}
\end{array}
\end{equation}
and
\begin{equation}\label{tau_r}
\begin{array}{lll}
P(\tau_r =N_{i+1}\mid \mathcal{G_{\infty}}) &=&\left(\displaystyle \prod_{j=0}^{i-1} P^{**}_X(N_j, N_{j+1})\cdot P^*_X(N_i, N_{i+1})\right)_{S_0, K}
\end{array}
\end{equation}
and
\begin{equation}\label{diff}
\begin{array}{lll}
 &&P(\tau_r-\tau_e > t \mid \mathcal{G_{\infty}})\\
&=&\displaystyle \sum_{i=0}^{\infty}\left(\prod_{j=0}^{i-1} P^{**}_X(N_j, N_{j+1})\cdot P^*_X(N_i, N_{i+1}-t)\right)_{S_0, K} \exp\left\{ -\int_{N_{i+1}-t}^{N_{i+1}} \lambda_K(X_u) du\right\}
\end{array}
\end{equation}
where conditioning on the underlying process $(X_t)_{t \geq 0}$,
 $P_{X}(s,t)$ denotes the transition probability matrix of the state process $(S_t)_{t \geq 0}$, i.e., the $(i,j)$ entry of $P_{X}(s,t)$
denotes the probability that the firm stays in state $j$ at time $t$ given that
the firm stays in state $i$ at time $s$.
$P^{*}_X(s, t)$ is the $(K-1)\times K$ matrix that results from deleting the $K$th row of $P_{X}(s,t)$ and $P^{**}_X(s, t)$ is the $(K-1)\times (K-1)$ matrix that results from deleting the $K$th column and $K$th row of $P_{X}(s,t)$.
\end{proposition}

\begin{proof}
See Appendix A.
\end{proof}

From Proposition \ref{Prop1}, one can see that the probability law of $\tau_r$ and $\tau_e$
depends on the transition matrix $P_X(s, t)$.
In the following, we discuss the issue of calculating  $P_X(s, t)$ in different cases.

\subsection{Constant Transition Rates}

In this subsection, we assume that the underlying stochastic process is degenerate,
which means that $X_u=c, u \geq 0$ for some constant $c$.
Let $\lambda_{i, j}(c)=\lambda_{i, j}$ and $\lambda_i(c)=\lambda_i$
for all $i, j$ and $P_X(s, t)=P(s, t)$. Let
$$
A=\left(
\begin{array}{cccccccccccc}
			-\lambda_1& \lambda_{1,2}& \lambda_{1,3}&\ldots& \ldots &\lambda_{1,K}\\
			\lambda_{2,1}& -\lambda_2& \lambda_{2,3}&\ldots& \ldots &\lambda_{2,K}\\
			\lambda_{3,1}& \lambda_{3,2}& -\lambda_{3}&\ldots& \ldots &\lambda_{3,K}\\
			\vdots            &\vdots      &\ddots  &\ddots &\vdots             &\vdots\\
	\lambda_{K-1,1}& \lambda_{K-1,2}&\ldots& \ldots &-\lambda_{K-1}&\lambda_{K-1,K}\\
			\lambda_{K,1}& \lambda_{K,2}&\ldots& \ldots &\lambda_{K,K-1}&-\lambda_{K}\\
\end{array}\right)
$$
then by Kolmogorov's backward equations, one can obtain
\begin{equation}\label{de}
\frac{\partial P(s,t)}{\partial s}=-A P(s,t).
\end{equation}
Solving these equations, we obtain
$$
P(s,t) = \exp \left( A (t-s) \right).
$$
In the following, we give an example of two states.

\begin{example}
In this example, we assume that the firm's state process
follows a two-state continuous-time Markov chain
with normal state ``$1$'' and default state ``$2$''.
 The transition rate is given by $\lambda_1$ and $\lambda_2$,
hence
$$
A=\left(
\begin{array}{cc}
-\lambda_1 & \lambda_1\\
\lambda_2 & -\lambda_2
\end{array}\right)
$$
and
$$
P(s, t)=\left(
\begin{array}{cc}
\frac{\lambda_1}{\lambda_1+\lambda_2}e^{-(\lambda_1+\lambda_2)(t-s)} +\frac{\lambda_2}{\lambda_1+\lambda_2} &-\frac{\lambda_1}{\lambda_1+\lambda_2}e^{-(\lambda_1+\lambda_2)(t-s)} +\frac{\lambda_1}{\lambda_1+\lambda_2}\\
-\frac{\lambda_2}{\lambda_1+\lambda_2}e^{-(\lambda_1+\lambda_2)(t-s)} +\frac{\lambda_2}{\lambda_1+\lambda_2}&
 \frac{\lambda_2}{\lambda_1+\lambda_2}e^{-(\lambda_1+\lambda_2)(t-s)} +\frac{\lambda_1}{\lambda_1+\lambda_2}
\end{array}\right)
$$
By Proposition 1, one obtains
\begin{equation}\label{e11}
P(\tau_e \in (N_i, N_i +t] )=\left(\frac{\lambda_1}{\lambda_1+\lambda_2}e^{-(\lambda_1+\lambda_2)N} +\frac{\lambda_2}{\lambda_1+\lambda_2}\right)^i
\left(\frac{\lambda_1}{\lambda_1+\lambda_2}
-\frac{\lambda_1}{\lambda_1+\lambda_2}e^{-(\lambda_1+\lambda_2)t}\right)e^{-\lambda_2(N-t)}.
\end{equation}
and
\begin{equation}\label{e22}
P(\tau_r =N_{i+1} )=\left(\frac{\lambda_1}{\lambda_1+\lambda_2}e^{-(\lambda_1+\lambda_2)N} +\frac{\lambda_2}{\lambda_1+\lambda_2}\right)^i
\left(-\frac{\lambda_1}{\lambda_1+\lambda_2}e^{-(\lambda_1+\lambda_2)N} +\frac{\lambda_1}{\lambda_1+\lambda_2}\right)
\end{equation}
and
\begin{equation}\label{e33}
\displaystyle
P(\tau_r-\tau_e > t)=\frac{e^{-\lambda_2 t}-e^{-(\lambda_1+\lambda_2)N}e^{\lambda_1 t}}
{1-e^{-(\lambda_1+\lambda_2)N} } .
\end{equation}
\end{example}

\subsection{Stochastic Transition Rates}

We define the following matrix
$$
A_X(s)=\left(
\begin{array}{cccccccccccc}
			-\lambda_1(X_s)& \lambda_{1,2}(X_s)& \lambda_{1,3}(X_s)&\ldots& \ldots&\lambda_{1,K}(X_s)\\
			\lambda_{2,1}(X_s)& -\lambda_2(X_s)& \lambda_{2,3}(X_s)&\ldots & \ldots &\lambda_{2,K}(X_s)\\
			\lambda_{3,1}(X_s)& \lambda_{3,2}(X_s)& -\lambda_{3}(X_s)&\ldots& \ldots& \lambda_{3,K}(X_s)\\
			\vdots            &\vdots        &\ddots  &\ddots &\vdots             &\vdots\\
			\lambda_{K-1,1}(X_s)& \lambda_{K-1,2}(X_s)&\ldots & \ldots &-\lambda_{K-1}(X_s)&\lambda_{K-1,K}(X_s)\\
			\lambda_{K,1}(X_s)& \lambda_{K,2}(X_s)&\ldots& \ldots& \lambda_{K,K-1}(X_s)&-\lambda_{K}(X_s)\\
\end{array}\right)
$$
and we obtain
\begin{equation}\label{de2}
\frac{\partial P_{X}(s,t)}{\partial s}=- A_X(s) P_{X}(s,t).
\end{equation}
As shown in Lando (1998) \cite{Lando}, in general,
$$
P_{X}(s,t) \neq \exp \left[\int_s^t A_X(u)du\right].
$$
Hence we adopt the special structure of $A_X(s)$ in Lando (1998) \cite{Lando}
by assuming that
$$A_X(s) =B \mu(X_s)B^{-1},$$
where
$\mu(X_s)$ denotes the $K\times K$ diagonal matrix
$$
{\rm diag} (\mu_1(X_s), \ldots, \mu_{K-1} (X_s), \mu_K(X_s))
$$
with $\mu_K(X_s)=0$, and $B$ denotes the $K \times K$ matrix whose columns consist of $K$ eigenvectors of $A_X(s)$. Let
$$
E_X(s, t) = {\rm diag} \left(\exp \left[\int_s^t \mu_1(X_u) du\right], \ldots, \exp \left[\int_s^t \mu_{K-1}(X_u) du \right], \exp \left[\int_s^t \mu_K(X_u) du\right]\right).
$$
Then one can obtain the following lemma.

\begin{lemma}
We have
$$
P_X(s, t ) = B E_X(s, t) B^{-1}
$$
satisfying 
Eq. (\ref{de2}) and is the desired transition probability matrix.
\end{lemma}
\begin{proof}
By using the similar argument in Lando (1988) \cite{Lando}.
\end{proof}

\subsubsection{An Affine Jump Diffusion Model for $(X_s)_{s \geq 0}$}

In this subsection, we adopt an affine jump diffusion process to characterize the dynamics of $(X_s)_{s \geq 0}$.
As we know, the basic affine process is attractive in modeling credit risk for its tractability,  see for instance Duffie and Kan (1996) \cite{DK} and
Duffie and G$\hat{a}$rleanu (2001) \cite{DG}.
We assume that
\begin{equation}\label{dyn}
d X_t=\kappa(\theta-X_t)dt + \sigma \sqrt{X_t} dB_t + d J_t
\end{equation}
where $B_t$ is a standard Brownian motion and $$J_t=\sum_{i=1}^{N(t)}Z_i$$
with $N(t)$ being counting jumps in Poisson with intensity $\lambda$ and $\{Z_i\}_{i=1}^{\infty}$ a sequence of i.i.d. exponentials with mean $\gamma$.
Then the expectation
\begin{equation}\label{exp}
E\left[e^{\int_t^T RX_u du +w X_T} \mid \mathcal{G}_t\right]=e^{\alpha(T-t; R,w)+\beta(T-t; R,w)X_t},
\end{equation}
where $R, w$ are constants and $\alpha, \beta$ are coefficient functions satisfying the ODEs
$$
\left\{
\begin{array}{lll}
\displaystyle \frac{d{\alpha}(s; R,w)}{ds}&=& \displaystyle \kappa \theta {\beta}(s; R,w)+\frac{\lambda \gamma {\beta}(s; R,w)}{1-\gamma {\beta}(s; R,w)}\\
\displaystyle \frac{d{\beta}(s; R,w)}{ds} &=& \displaystyle -\kappa {\beta}(s; R,w) +\frac{1}{2}\sigma^2 {\beta}(s; R,w)^2(s)+R
\end{array}
\right.
$$
with $\alpha(0; R,w)=0$ and $\beta(0; R,w)=w$.
The explicit form of $\alpha(s; R,w)$ and $\beta(s; R,w)$
is given by Duffie and G$\hat{a}$rleanu (2001) \cite{DG}.
The solution to $\beta(s; R, w)$ is given by
$$
\beta(s; R, w)=\frac{1+a e^{b s}}{c+d e^{b s}}
$$
where the coefficients depend on $R$ and $w$,
$$
\left\{
\begin{array}{lll}
a&=&(d+c)w-1\\
b&=& \displaystyle \frac{d(-\kappa +2 Rc)+a(-\kappa c+\sigma^2)}{ac-d}\\
c&=& \displaystyle \frac{\kappa+\sqrt{\kappa^2-2 R \sigma^2}}{2 R}\\
d&=& \displaystyle (1-cw)\frac{-\kappa+\sigma^2 w+\sqrt{(-\kappa+\sigma^2 w)^2-\sigma^2}}{-2\kappa w+\sigma^2 w^2 +2R}
\end{array}
\right.
$$
and $\alpha(s; R, w)$ follows from solving the ODE by substituting $\beta(s; R,w)$.

In what follows, we implement the calculation of distribution
of $\tau_e$ and $\tau_r$ given the dynamics of $(X_s)_{s \geq 0}$ as in
Eq. (\ref{dyn}).
We assume that $\mu_i(X_s)=\mu_i X_s$ with $\mu_i$ being a constant for $i=1, 2, \ldots, K-1$, and $\mu_K=0$.
Although the computational method works in multi-state case,
here for simplicity of discussion, we assume that $K=2$, i.e.,
the operation state of a firm either ``normal'' or ``default''.
Before we state the main result of this subsection,
we have the following observations:
$$
P^{**}_X(s,t)=B^* E_X(s,t) B^{-1}_*
\quad {\rm and} \quad
P^{*}_X(s,t)=B^* E_X(s,t) B^{-1}
$$
where $B^*$ denotes the $(K-1)\times K$ matrix that
results from deleting the $K$th row of $B$,
$B^{-1}_*$ denotes the $K\times (K-1)$ matrix
that results from deleting the $K$th column of $B^{-1}$.
When $K=2$,
$$
P^{**}_X(s,t)=m_1\exp \left[\int_s^t \mu_1(X_u) du\right] +m_2 \exp \left[\int_s^t \mu_2(X_u) du\right]
$$
where $m_1=b_{11}b^{(-1)}_{11}$ and $m_2=b_{12}b^{(-1)}_{21}$ with $b_{i,j}=B_{i,j}$ and $b^{(-1)}_{ij}=B^{-1}_{i,j}$.
We have
$$
\begin{array}{lll}
&&P^{*}_X(s,t)\\
&=&\left(m_1\exp [\int_s^t \mu_1(X_u) du] +m_2 \exp [\int_s^t \mu_2(X_u) du],
n_1\exp [\int_s^t \mu_1(X_u) du] +n_2 \exp [\int_s^t \mu_2(X_u) du]\right)
\end{array}
$$
where $n_1=b_{11}b^{(-1)}_{12}$ and $n_2=b_{12}b^{(-1)}_{22}$.
And
$$
\lambda_2(X_s)=-p_1 \mu_1(X_u) -p_2  \mu_2(X_u)
$$
where $p_1=b_{21}b^{(-1)}_{12}$ and $p_2=b_{22}b^{(-1)}_{22}$.
Let
$$
\hat{E}_i:=\{{\bf e}=(e_0,e_1, \ldots, e_{i}): e_k \in \{1, 2\}\}.
$$
For each ${\bf e} \in \hat{E}_i$, let
$$
\hat{m}({\bf e})=n_{e_i}\prod_{j=0}^{i-1}m_{e_j}
$$
and
$$
\begin{array}{lll}
\hat{\mu}({\bf e}, s) &=&1_{\{s \in[N_i+t, N_{i+1})\}}[p_1\mu_1(X_s)+p_2\mu_2(X_s)]+1_{\{s \in[N_i, N_i+t)\}}\mu_{e_i}(X_s) \\
&&+ \displaystyle \sum_{j=0}^{i-1} 1_{\{s \in[N_j, N_{j+1})\}}\mu_{e_j}(X_s).
\end{array}
$$

\begin{proposition}\label{Prop3}
If $K=2$, $\mu_i(X_s)=\mu_i X_s$ with $\mu_1$ being a constant, $\mu_2=0$,
the distribution of $\tau_e$  is given by
\begin{equation}\label{1st}
P(\tau_e \in(N_i, N_i+t])=\sum_{{\bf e}
\in \hat{E}_i} \hat{m}({\bf e})\left(\prod_{j=0}^{i+1} v_j({\bf e})\right) \exp[\beta(N; R_0({\bf e}), w_0({\bf e})) X_0]
\end{equation}
where $R_j, w_j, v_j$ are defined in Appendix B.1.
The distribution of $\tau_r$ and the difference $\tau_r-\tau_e$ are given by,
\begin{equation}\label{2nd}
P(\tau_r =N_{i+1}) = P(\tau_e \in (N_i, N_{i+1}] ).
\end{equation}
and
\begin{equation}\label{3rd}
P(\tau_r-\tau_e > t ) = \sum_{i=0}^{\infty} P(\tau_e \in (N_i, N_{i+1}-t] ).
\end{equation}
\end{proposition}

\begin{proof}
See Appendix B.1.
\end{proof}

We note that when conducting the numerical experiment,
we apply Eq. (\ref{3rd}) to approximate $P(\tau_r-\tau_e >t)$,
where the error is given by
$$
\left|P(\tau_r-\tau_e>t)-\sum_{i=0}^k P(\tau_e\in(N_i,N_{i+1}-t])\right|
< P(\tau_r > N_{k+1}) \rightarrow 0
$$
as $k\rightarrow \infty$.
For the ease of computing the probability $P(\tau_e \in(N_i, N_i+t])$, we establish the following.
\begin{proposition}
If $K=2$, $\mu_i(X_s)=\mu_i X_s$ with $\mu_1$ being a constant, $\mu_2=0$,
the distribution of $\tau_e$  is given by
\begin{equation}\label{1st1}
P(\tau_e \in(N_i, N_i+t])=\sum_{j=1}^{2^{i+1}}a_{i,j}\exp(b_{i,j}X_0),
\end{equation}
where
$$
\begin{array}{rcl}
a_{i+1,j}&=&\left\{
\begin{array}{ll}
m_1 a_{i,j} \exp(\alpha(N,\mu_1,b_{i,j})),& j=1,2,\ldots, 2^{i+1}\\
m_2 a_{i, j-2^{i+1}} \exp(\alpha(N,\mu_2,b_{i,j-2^{i+1}})),& j=2^{i+1}+1,2^{i+1}+2,\ldots, 2^{i+2}
\end{array}
\right.
\\[3mm]
b_{i+1,j}&=&\left\{\begin{array}{ll}
\beta(N,\mu_1,b_{i,j}), &j=1,2,\ldots, 2^{i+1}\\
\beta(N,\mu_2,b_{i,j-2^{i+1}}), &j=2^{i+1}+1,2^{i+1}+2,\ldots, 2^{i+2}
\end{array}
\right.
\end{array}
$$
and 
$$\left\{
\begin{array}{l}
a_{0,1}=n_1 \exp[\alpha(N-t, p_1\mu_1+p_2\mu_2, 0)\alpha(t, \mu_1, \beta(N-t, p_1\mu_1+p_2\mu_2, 0))]\\
a_{0,2}=n_2 \exp[\alpha(N-t, p_1\mu_1+p_2\mu_2, 0)\alpha(t, \mu_2, \beta(N-t, p_1\mu_1+p_2\mu_2, 0))]\\
b_{0,1}=\beta(t, \mu_1, \beta(N-t, p_1\mu_1+p_2\mu_2, 0))\\
b_{0,2}=\beta(t, \mu_2, \beta(N-t, p_1\mu_1+p_2\mu_2, 0)).
\end{array}\right.
$$
\end{proposition}
\begin{proof}
See Appendix B.2.
\end{proof}

\section{Numerical Experiments and Discussions}

{In this section, we first discuss the constant intensity rate model.
The model parameters can be solved by employing the maximum likelihood approach.
We state the sufficient conditions for the density function to have a
``$U$-shape''.
Numerical results are then given to demonstrate the model.
However, the constant intensity model does not fit the real data very
well though it has the ``$U$-shape'' property.
We then present the numerical results for the stochastic intensity 
intensity model.
It is found by varying the parameters $\kappa, \gamma$ and $\sigma$,
different ``$U$-shape'' density functions can be obtained.
Thus it is clear that the stochastic intensity rate model can better
fit the real data as it includes the constant rate intensity model
as its particular case.

In the stochastic intensity rate model, 
we note that if the mean-reverting rate $\kappa$ is getting
large, the effect of stochastic part will be diminished.
Eventually the process will be dominated by deterministic part
$d X_t=\kappa (\theta - X_t)dt$.
The parameter $\kappa$ characterized the internal factor of
the firm default process.
One expects that when $\kappa$ increases,
the distribution seems to converge to certain ``$U$-shape'' function
and this is consistent with the results in Figure 2.

The parameter $\gamma$, 
the mean jump size of the jump process $J_t$
which is a positive quantity, can be regarded as the severity of
an external event causing the stress.
We remark that sign of the jump is always positive.
The larger the value is, the more likely that the time lap between the
economic default time and the recorded time is short.
Thus we expect that  when $\gamma $ increases,
the distribution will have a flatter and flatter tail and this
is consistent with the results in Figure 3.

Finally, the non-negative parameter $\sigma$ controls 
the effect of the stochastic part of a Brownian motion $\sigma dW$ 
which can be positive or negative and it 
represents the external market risk.
We expect that when $\sigma$ increases,
the better capital-structured companies have larger
time gap between the economic and the recorded default 
while worse capital-structured companies have 
shorter time gap between those. 
The impact of increasing $\sigma$ on both type of companies reveals in the time difference of the two default times as in Figure 4.

In a more economic sense, the parameter $\sigma$ can be interpreted as a measure of degree the macroeconomic fluctuation or market condition.
The larger $\sigma$ is, the more firms are to default 
given their original status.
As shown by Jacobson et al (2011) \cite{Jacobson}, strong evidence for a substantial and stable impact from aggregate fluctuations and business defaults are found in large banking crisis.
Moreover, default frequencies tend to increase significantly when the economy fluctuates more.
Intuitively speaking, when market conditions or macroeconomy becomes more uncertain or worse, bank or other lenders tend to be less confident
and retract their lending to firms, making firms more easily to default.
Another interesting facts about our model is that there is a ``shift'' in the distribution of firms' ``default gap classes''.
Comparing the first and the third graph in Figure 4,
it is not hard to see that the distribution of firms' default gap
tends to shift along the parabola rightwards,
lifting the right tail up while pressing the left tail down.
Moreover, it is obvious that the shifts from the classes with
larger default gap are bigger than those from the class
with smaller default gap.
This interesting phenomenon can be interpreted in a very reasonable way.
It is known to all that firms' capital structure and governance manner etc.
are very important measure of firms' strength.
In particular, these properties tend to be more variable
or of larger variance in start-up firms or less matured firms.
Baek et al. (2004) \cite{Baek}
and Ivashina and Scharfstein (2008) \cite{Ivashina}
found that firms with better governance manner and capital structure are more likely to survive from defaults during crisis.
Start-up firms or less-developed firms (lower class firms)
systematically have larger default
gaps than those larger and matured firms (higher class firms).
Good candidates in the lower class, namely those firms less-matured,
but with relatively better governance manner or reasonable capital structures, will have better access to funding or lending during crisis compared with their peers in the same class.
We expect the good candidates in each classes that are making the shift.
And the shift magnitudes are larger in the lower classes because
the variance of capital structure and governance manner are larger in these lower classes.}

\subsection{Constant Intensity}

In this section, we first present some estimation method for
solving the model parameters.
We then compare our proposed model with the real data
extracted from Guo, Jarrow and Lin (2008) \cite{Guo1}.
For the real data, Table 1 reports the time difference between the economic and recorded default date with $N=180$ days, extracted from Guo, Jarrow and Lin (2008) \cite{Guo1}.
From the table, one can easily observe that  the density function of the time
difference between the economic and recorded default time
has a ``$U$-shape''.

Regarding our model, we  assume the state process follows the two-state continuous-time Markov chain as in Example 1.
Indeed, from Eq. (\ref{e33}), we observe that
the density function of the time difference between the economic and
recorded default time is always convex.
In fact, it can be shown easily that

\begin{lemma}
The density function has a ``$U$-shape'' behavior
as long as the following conditions are satisfied:\\
(i) $ e^{-(\lambda_1+\lambda_2)N/2} \lambda_1-\lambda_2 \leq 0$\\
(ii) $0 \le \lambda_1-\lambda_2$.\\
\end{lemma}

We remark that if $N$ is large, then
$e^{-(\lambda_1+\lambda_2)N/2} \approx 0$ and therefore
essentially the sufficient condition in the above lemma
will become $\lambda_2 \le \lambda_1$.

To estimate the model parameters,
we adopt the Maximum Log-likelihood method to estimate the desired parameter $\lambda_1$ and $\lambda_2$ (see Appendix C),
from which we obtain the estimate of the two parameters:
$$
\lambda_1=0.3631 \quad {\rm and} \quad \lambda_2=0.0238.
$$
We also present the density function of the time difference between the
economic and recorded default time with comparison of
the proposed model (Example 1) and the real data.
We note that
the two-state constant rate model does not fit the
real data very well though it can capture the important `
`$U$-shape property'' of the distribution.

\begin{table}
\caption{Time between the economic and recorded default dates}
\label{tab:1}
\begin{center}
\begin{tabular}{c c c c c c }
\hline\noalign{\smallskip}
  {\bf Day}       & $(0,18]$& $(18,36]$ & $(36,54]$ & $(54,72]$ & $(72,90]$\\
  {\bf Number of Firms} & 24    & 13      & 6       & 5       & 3       \\
\noalign{\smallskip}\hline\noalign{\smallskip}
 {\bf Day}   &$(90,108]$ & $(108,126]$& $(126,144]$& $(144,162]$ & $(162,180]$ \\
 {\bf Number of Firms}   &        1 & 4          & 4           & 2           & 11        \\
\noalign{\smallskip}
\hline
\end{tabular}
\end{center}
\end{table}

\begin{figure}
	\centering
		\resizebox{12cm}{!}{\includegraphics{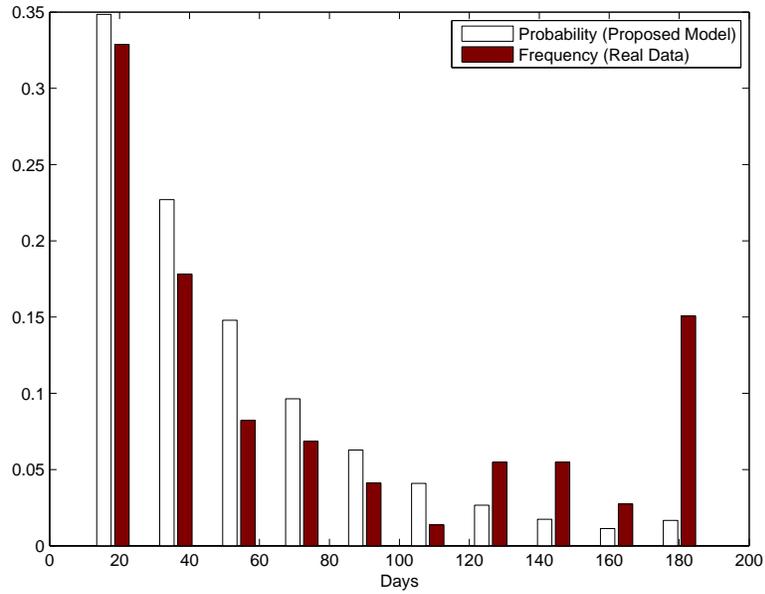}}
			\caption{A comparison of the two-state constant rate
			model and the real data.}
\end{figure}

\subsection{Stochastic Intensity}

In this example, we assume that the state process of the firm
$(S_t)_{t \geq 0}$ follows a two-state continuous-time Markov chain
with stochastic transition rates depending on the underlying process
$(X_t)_{t \geq 0}$ as described in Section 3.2.1.
By setting
$$
\mu_1=-0.52, \quad \mu_2=0, \quad  \theta=1, \quad \lambda=0.2, \quad  X_0=1, \quad N=180
$$
and
$$
B=\left(
\begin{array}{cc}
-0.9992&-0.7071\\
0.0400&-0.7071
\end{array}
\right).
$$
and vary the value of parameters $\kappa$, $\gamma$ and $\sigma$,
we compute the density function of the time between the recorded and
the economic default in Figures 2, 3 and 4.
{By setting parameters as above, the initial state is $A_X$ is given by
$$
A_X(0)=\left(
\begin{array}{cc}
-0.5000 & 0.5000\\
 0.0200 &-0.0200
\end{array}
\right).
$$ }
Figure 3 shows that as the jump size increase,
which means that the common factor suffers from a larger jump,
the difference of the two default time tends to decrease.
We demonstrate in Figure 4 that,
as the volatility of the common factor decrease,
the difference of the default times increases.

\begin{figure}
	\centering
		\resizebox*{12cm}{!}{\includegraphics{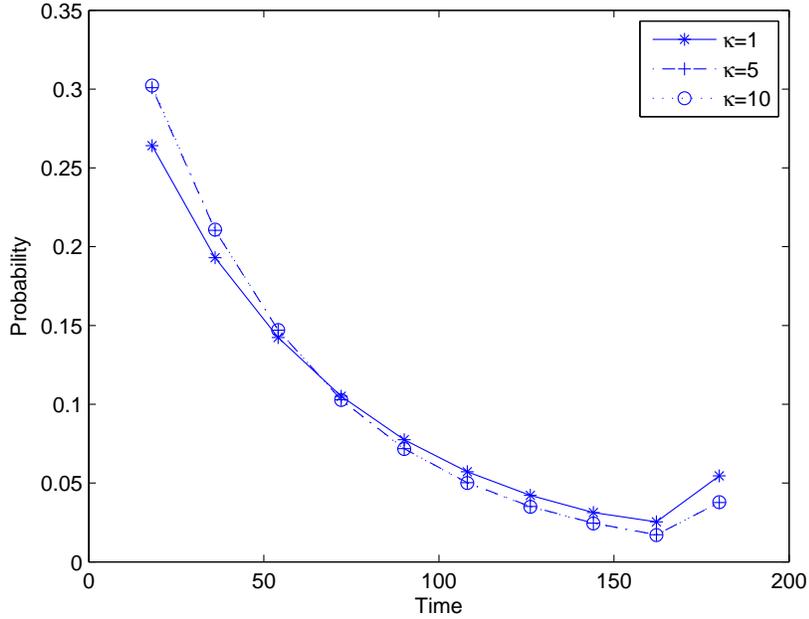}}
			\caption{Distribution of time difference between economic and recorded default with $\sigma=5, \gamma=0.1$ and different $\kappa$.}
\end{figure}

\begin{figure}
	\centering
		\resizebox*{12cm}{!}{\includegraphics{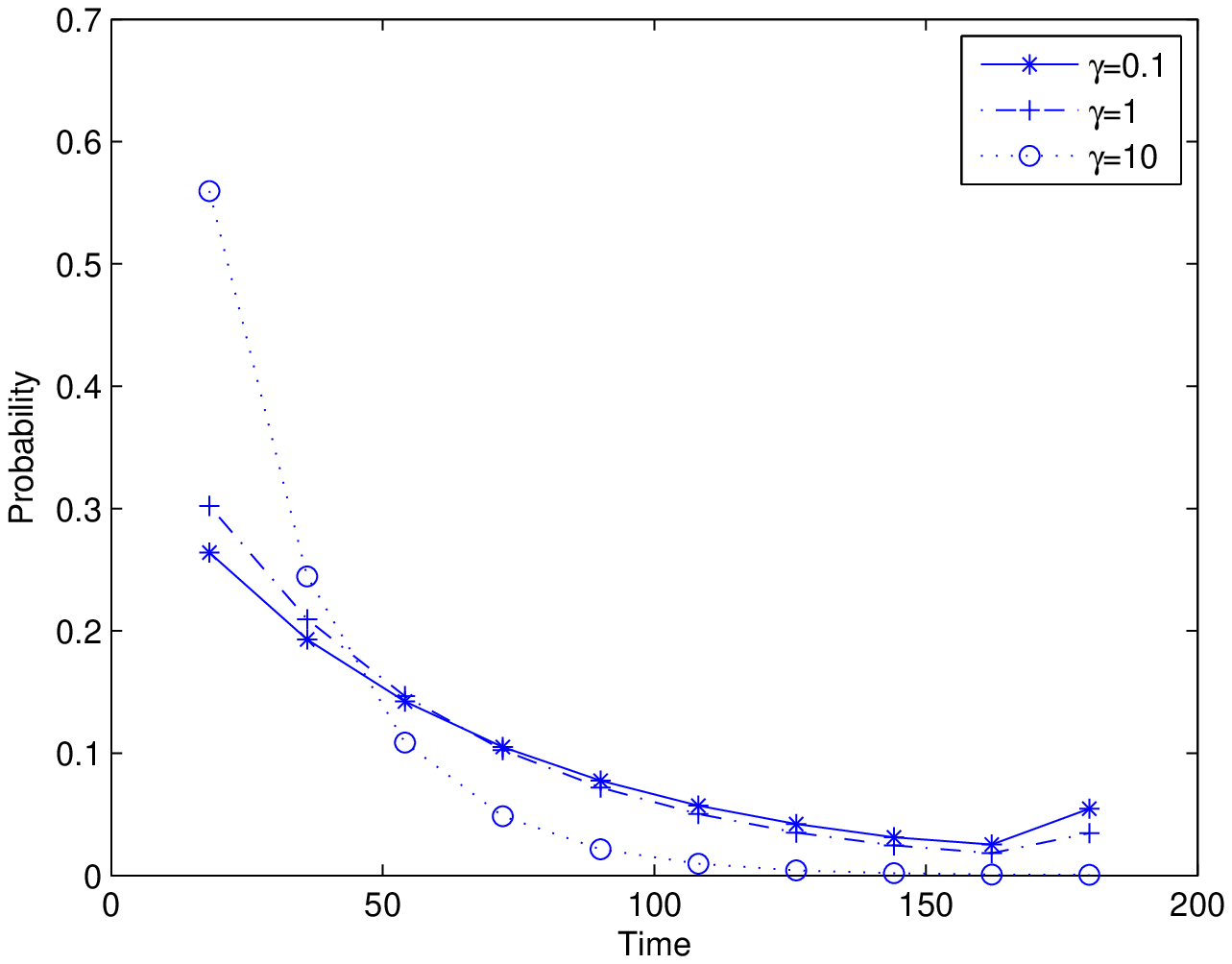}}
			\caption{Distribution of time difference between economic and recorded default with $\kappa=1, \sigma=5$ and different $\gamma$.}
\end{figure}

\begin{figure}
	\centering
		\resizebox*{12cm}{!}{\includegraphics{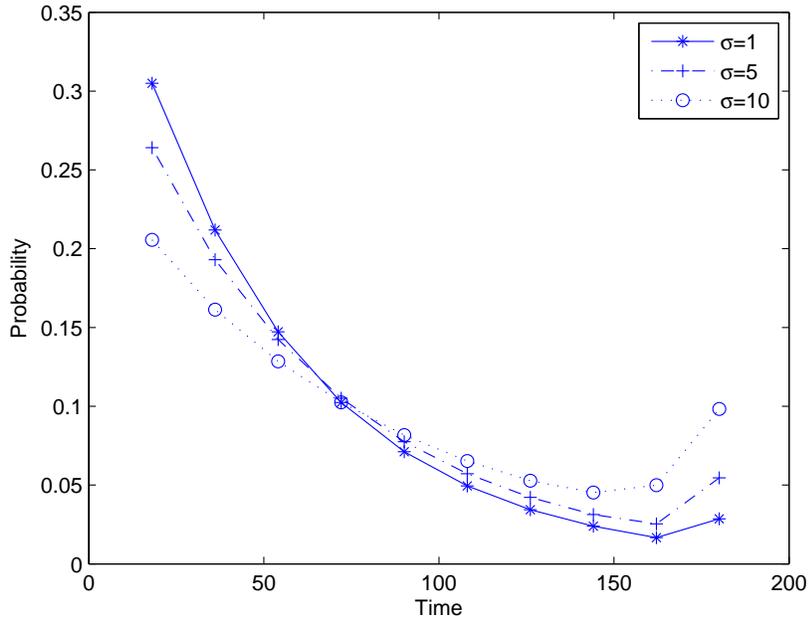}}
			\caption{Distribution of time difference between economic and recorded default with $\kappa=1, \gamma=0.1$ and different $\sigma$.}
\end{figure}

{
For the two-state stochastic transition rate model,
again we present the distribution of time difference between economic
and recorded default in Figure 5.
We assume the parameters are given by
$$
\mu_1=-0.5120,\mu_2=0,   \theta=1, \lambda=0.2, \kappa=1,  \sigma=9, \gamma=3.6,  X_0=1, N=180
$$
and
$$
B=\left(
\begin{array}{cc}
-0.9997&-0.7071\\
0.0246&-0.7071
\end{array}
\right),
$$
where the initial state of $A_X$ is given by
$$
A_X(0)=\left(
\begin{array}{cc}
-0.5000&0.5000\\
 0.0120&-0.0120
\end{array}
\right).
$$
The above set of parameters are obtained by
performing a grid search on $\kappa,\sigma$ and $\gamma$
with the object of  minimizing the mean squares of errors.
Therefore the two-state stochastic rate model fits the real data
quite well.

\begin{figure}
	\centering
		\resizebox*{12cm}{!}{\includegraphics{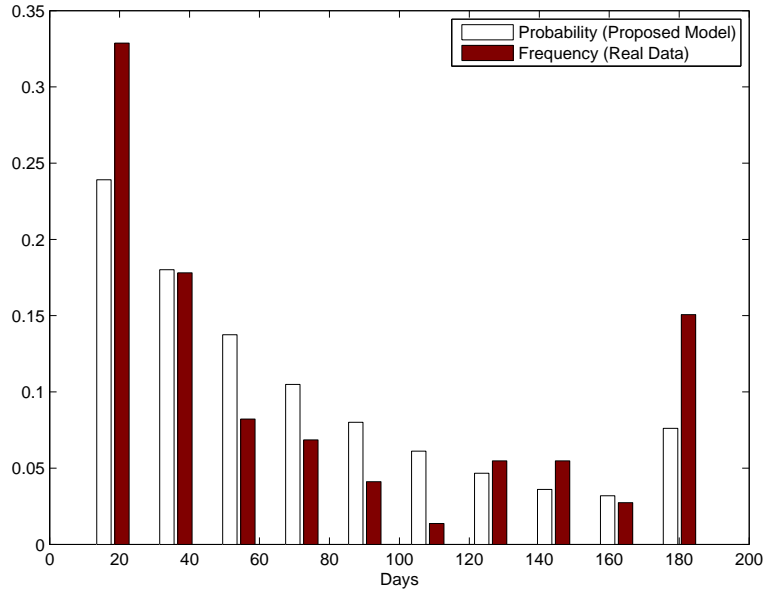}}
			\caption{Distribution of time difference between economic and recorded default with stochastic intensity.}
\end{figure}}

\section{Concluding Remarks}

In this paper, we develop a reduced-form model to
characterize the economic and recorded default time.
We assume the state process follows the continuous-time
Markov chain with stochastic transition rates depending
on the macroeconomic common factor.
We derive the probability law of $\tau_e$ and $\tau_r$ which depend on the stochastic transition matrix $P_X(s,t)$.
We also present the evaluation of $P_X(s,t)$ in different cases.
We investigate the probability distribution of the economic and recorded default time with constant transition rates and also with underlying common factor following basic affine jump diffusion.
Numerical experiments show that our proposed model can capture the features of empirical data.

The two-state constant rate model can capture the ``$U$-shape'' property
but the real data does not fit the model.
For our future research,  we shall consider a
multi-state constant rate model.
We expect the introduction of extra states can help to improve the
model and hence better fit the real data.
Regarding two-state stochastic rate model, 
we applied grid search method to obtain the model parameters.
We shall develop estimation method for the model parameters in our
future research.

\section{Appendix}

\subsection{Appendix A (Proof of Proposition \ref{Prop2})}

We note that Eq. (\ref{tau_r}) follows from Eq. (\ref{tau_e}) by using
$$
P(\tau_r =N_{i+1}\mid \mathcal{G_{\infty}})=P(\tau_e \in (N_i, N_{i+1}] \mid \mathcal{G_{\infty}}).
$$
Eq. (\ref{diff}) follows from Eq. (\ref{tau_e}) by using
$$
P(\tau_r-\tau_e > t \mid \mathcal{G_{\infty}})=\sum_{i=0}^{\infty} P(\tau_e \in (N_i, N_{i+1}-t] \mid \mathcal{G_{\infty}}).
$$
And Eq. (\ref{tau_e}) follows by
$$
\begin{array}{lll}
& &P(\tau_e \in (N_i, N_i +t] \mid \mathcal{G_{\infty}})\\
&=&\displaystyle \sum_{n_i=1}^{K-1}P(S_{N_1} \neq K, \ldots, S_{N_{i-1}} \neq K, S_{N_i} =n_i \mid \mathcal{G_{\infty}})P(\tau_e \in (N_i, N_i +t] \mid S_{N_i} =n_i, \mathcal{G_{\infty}})\\
&=&\displaystyle \sum_{n_i=1}^{K-1}\sum_{n_{i-1}=1}^{K-1}\ldots \sum_{n_{1}=1}^{K-1}P(S_{N_1} =n_1, \ldots, S_{N_{i-1}} =n_{i-1}, S_{N_i} =n_i \mid \mathcal{G_{\infty}})\\
&& \times P_X(N_i, N_i +t)_{n_i, K}\exp\left\{ -\int_{N_i+t}^{N_{i+1}} \lambda_K(X_u) du\right\}\\
&=&\displaystyle \sum_{n_i=1}^{K-1}\sum_{n_{i-1}=1}^{K-1}\ldots \sum_{n_{1}=1}^{K-1}P_X(N_0, N_1)_{S_0, n_1}\ldots P_X(N_{i-1}, N_i)_{n_{i-1}, n_i}  P_X(N_i, N_i +t)_{n_i, K}\\
&& \times\exp\left\{ -\int_{N_i+t}^{N_{i+1}} \lambda_K(X_u) du\right\}\\
&=&\left(\displaystyle \prod_{j=0}^{i-1} P^{**}_X(N_j, N_{j+1})\cdot P^*_X(N_i, N_i +t)\right)_{S_0, K}\exp\left\{ -\int_{N_i+t}^{N_{i+1}} \lambda_K(X_u) du\right\}.
\end{array}
$$

\subsection{Appendix B.1( Proof of Proposition \ref{Prop3})}
\begin{proof}
Eq.s (\ref{2nd}) and (\ref{3rd}) are obvious and
it suffices to show  Eq. (\ref{1st}).
Now we have
$$
\begin{array}{lll}\label{taue}
P(\tau_e \in(N_i, N_i+t]\mid \mathcal{G_{\infty}})
&=&\displaystyle \prod_{j=0}^{i-1}\left(m_1\exp \left[\int_{N_j}^{N_{j+1}} \mu_1(X_u) du\right] +m_2 \exp \left[\int_{N_j}^{N_{j+1}} \mu_2(X_u) du\right]\right)\\
&& \displaystyle \times \left(n_1\exp \left[\int_{N_{i}}^{N_i+t} \mu_1(X_u) du\right] +n_2 \exp \left[\int_{N_{i}}^{N_i+t} \mu_2(X_u) du\right]\right)\\
&& \displaystyle \times \exp \left[\int_{N_{i}+t}^{N_{i+1}} p_1\mu_1(X_u)+p_2\mu_2(X_u)du\right]\\
&=& \displaystyle \sum_{{\bf e} \in \hat{E}_i} \hat{m}({\bf e})\exp \left[\int_{N_0}^{N_{i+1}} \hat{\mu}({\bf e},u)du\right].
\end{array}
$$
Hence
\begin{equation}\label{ee}
P(\tau_e \in(N_i, N_i+t])
=\sum_{{\bf e} \in \hat{E}_i} \hat{m}({\bf e})
E\left(\exp\left[\int_{N_0}^{N_{i+1}} \hat{\mu}({\bf e},u)du\right]\right).
\end{equation}
For a fixed ${\bf e} \in \hat{E}_i$, let
$$
\begin{array}{lll}
R_{i+1}({\bf e})&=&p_1 \mu_1+p_2\mu_2\\
R_j({\bf e})&=&\mu_{e_j}, j=0,1,\ldots, i\\
w_{i+1}({\bf e})&=&0\\
w_i({\bf e})&=&\beta(N-t; R_{i+1}({\bf e}), w_{i+1}({\bf e}))\\
w_{i-1}({\bf e})&=&\beta(t; R_{i}({\bf e}), w_{i}({\bf e}))\\
w_j({\bf e})&=&\beta(N; R_{j+1}({\bf e}), w_{j+1}({\bf e})), j=0, 1, \ldots, i-2\\
v_{i+1}({\bf e})&=&\exp[\alpha(N-t; R_{i+1}({\bf e}), w_{i+1}({\bf e}))]\\
v_{i}({\bf e})&=&\exp[\alpha(t; R_{i}({\bf e}), w_{i}({\bf e}))]\\
v_{j}({\bf e})&=&\exp[\alpha(N; R_{j}({\bf e}), w_{j}({\bf e}))], j=0, 1, \ldots, i-1.\\
\end{array}
$$
Then we can rewrite $\hat{\mu}({\bf e}, s)$ as
$$
\hat{\mu}({\bf e}, s)=1_{\{s \in[N_i+t, N_{i+1})\}}(R_{i+1}({\bf e}) X_s)+1_{\{s \in[N_i, N_i+t)\}}(R_{i}({\bf e})X_s) +
\sum_{j=0}^{i-1} 1_{\{s \in[N_j, N_{j+1})\}}(R_j({\bf e}) X_s).
$$
Using the iterated expectation and Eq. (\ref{exp}) we obtain
$$
\begin{array}{lll}
& &E\left(\exp[\int_{N_0}^{N_{i+1}} \hat{\mu}({\bf e},u)du]\right)\\
&=&E\left(\exp[\int_{N_0}^{N_{i}+t} \hat{\mu}({\bf e},u)du]E(\exp[\int_{N_{i}+t}^{N_{i+1}} R_{i+1}({\bf e})X_u du] \mid \mathcal{G}_{N_i+t})\right)\\
&=&v_{i+1}({\bf e})E\left(\exp[\int_{N_0}^{N_{i}+t} \hat{\mu}({\bf e},u)du]  \exp[w_i({\bf e}) X_{N_i+t}]\right)\\
&=&v_{i+1}({\bf e}) E\left(\exp[\int_{N_0}^{N_{i}} \hat{\mu}({\bf e},u)du]E(\exp[\int_{N_{i}}^{N_{i}+t} R_{i}({\bf e})X_u du+w_i({\bf e}) X_{N_i+t}] \mid \mathcal{G}_{N_i})\right)\\
&=&v_{i+1}({\bf e})v_{i}(e) E\left(\exp[\int_{N_0}^{N_{i}} \hat{\mu}({\bf e},u)du] \exp[w_{i-1}({\bf e})X_{N_i}]\right)\\
&=&v_{i+1}({\bf e}) v_{i}({\bf e}) E\left(\exp[\int_{N_0}^{N_{i-1}} \hat{\mu}({\bf e},u)du]E(\exp[\int_{N_{i-1}}^{N_{i}} R_{i-1}({\bf e})X_u du+w_{i-1}({\bf e})X_{N_i}] \mid \mathcal{G}_{N_{i-1}})\right)\\
&=&v_{i+1}({\bf e}) v_{i}({\bf e}) v_{i-1}({\bf e}) E\left(\exp[\int_{N_0}^{N_{i-1}} \hat{\mu}({\bf e},u)du] \exp[w_{i-2}({\bf e})X_{N_{i-1}}]\right)\\
&=&(\prod_{j=0}^{i+1} v_j({\bf e})) \exp[\beta(N; R_0({\bf e}), w_0({\bf e})) X_0] \ {\rm (by \ iteration)}
\end{array}
$$
Hence Eq. (\ref{1st}) follows.
\end{proof}

\subsection{Appendix B.2( Proof of Proposition 4)}
\begin{proof}
We let
$$
H_i(X_0, t):=P(\tau_e \in(N_i, N_i+t])
$$
then by the proof of Proposition 3, for $i\geq 1$,
\begin{equation}\label{H}
\begin{array}{rcl}
P(\tau_e \in (N_i,N_i+t]\mid \mathcal{F}_{N_1})&=&\left(m_1\exp\left[\int_{N_0}^{N_1}\mu_1(X_u) du\right]+m_2\exp\left[\int_{N_0}^{N_1}\mu_2(X_u) du\right]\right)H_{i-1}(X_{N_1}, t)\\
H_i(X_0, t)&=&E\left[\left(m_1\exp\left[\int_{N_0}^{N_1}\mu_1(X_u) du\right]+m_2\exp\left[\int_{N_0}^{N_1}\mu_2(X_u) du\right]\right)H_{i-1}(X_{N_1}, t)\right]
\end{array}
\end{equation}
By Proposition 3, we obtain that
$$
H_0(x, t)=a_{0,1}\exp(b_{0,1}x)+a_{0,2}\exp(b_{0,2}x).
$$
Combining Eqs. (\ref{H}) and (\ref{exp}), Proposition 4 follows.
\end{proof}

\subsection{Appendix C}
Let $\delta=18$ days, $t_i=\delta i, i=0, 1, \ldots, 10$.
Let $N_i$ denote the number of firms whose time difference of economic and recorded default date is inside the interval  $(t_{i-1}, t_i]$.
Then the log-likelihood function is given by
$$
\begin{array}{lll}
\displaystyle \mathcal{L}(\lambda_1,\lambda_2) &=& \displaystyle \sum_{i=1}^{10} N_i \left(\ln \left[(e^{-\lambda_2 t_{i-1}}-e^{-\lambda_2 t_{i}})-e^{-(\lambda_1+\lambda_2)N}(e^{\lambda_1 t_{i-1}}-e^{\lambda_1 t_{i}})\right] \right. \\
&& \displaystyle \left. -\ln \left[1-e^{-(\lambda_1+\lambda_2)N} \right] \right)
\end{array}
$$
By setting
$$
\left\{
\begin{array}{lll}
\displaystyle \frac{\partial \mathcal{L}(\lambda_1,\lambda_2)}{\partial \lambda_1}=0\\
\displaystyle \frac{\partial \mathcal{L}(\lambda_1,\lambda_2)}{\partial \lambda_2}=0,
\end{array}
\right.
$$
we have two nonlinear equations for $\lambda_1$ and  $\lambda_2$.
Solving these equations numerically yields
$\lambda_1=0.3631$ and $\lambda_2=0.0238$.\\

\noindent
{\bf Acknowledgements:}
Research supported in part by GRF grants, HKU CERG grants and
HKU Hung Hing Ying Physical Research Grant.



\end{document}